\documentclass[11pt]{article}
\usepackage{a4wide}
\usepackage{hyperref}
\usepackage[usenames,dvipsnames]{color}
\usepackage{amssymb,amsmath,graphics,color,enumerate}
\usepackage{tikz}
\usepackage{xspace}
\usetikzlibrary{arrows,shapes,decorations,automata,backgrounds,petri,patterns} 
\usepackage{todonotes}
\usepackage[utf8x,utf8]{inputenc} 
\usepackage[T1]{fontenc}
\usepackage{booktabs}
\usepackage{array}
\usepackage{mathtools}
\usepackage{algorithm}
\usepackage{caption}
\usepackage[noend]{algpseudocode}
\usepackage[numbers,sort]{natbib}
\usepackage{standalone}
\usepackage{tikz}
\usetikzlibrary{arrows}
\usetikzlibrary{patterns}

\usepackage{epsfig}
\usepackage{subfig}
\usepackage{amsmath,amsfonts,amssymb,epsfig,color,amsthm}
\usepackage{comment}
\usepackage{thmtools}
\usepackage{thm-restate}
\usepackage{mdframed}
\usepackage[absolute]{textpos}
\usepackage{enumitem}

\newtheorem{definition}{Definition}[]
\newtheorem{theorem}{Theorem}[section]
\newtheorem{lemma}[theorem]{Lemma}

\newtheorem{fact}{Fact}

\newcommand{\Oh}{{\ensuremath{\mathcal{O}}}}

\newcommand{\Z}{\mathbb{Z}}
\newcommand{\F}{\mathbb{F}}

\newcommand{\anonymyze}[1]{}%

\newcommand{\ignore}[1]{}%

\newcommand{\DeepestDnCut}[1][]{{\rm DeepestDnCut}}
\newcommand{\DeepestDnCutNoMin}[1][]{{\rm DeepestDnCutNoMin}}
\newcommand{\DeepestDnCutNoMax}[1][]{{\rm DeepestDnCutNoMax}}
\newcommand{\LCA}[1][]{{\rm LCA}}
\newcommand{\td}{\mathrm{td}}
\newcommand{\tw}{\mathrm{tw}}
\newcommand{\tail}{\mathsf{tail}}
\newcommand{\tree}{\mathsf{tree}}
\newcommand{\chld}{\mathsf{chld}}
\newcommand{\comp}{\mathsf{comp}}
\newcommand{\depth}{\mathsf{depth}}
\newcommand{\cl}{\mathsf{cl}}
\newcommand{\im}{\mathrm{im}}
\newcommand{\LICZ}{\mathtt{CountElimTrees}}
\newcommand{\LICZLAS}{\mathtt{CountElimForest}}
\newcommand{\KONSTRUUJ}{\mathtt{ConstructElimForest}}
\newcommand{\ODZYSKAJ}{\KONSTRUUJ}

\newcommand{\ol}[1]{\overline{#1}}

\renewcommand{\setminus}{-}
\renewcommand{\leq}{\leqslant}

\renewcommand{\le}{\leqslant}
\renewcommand{\ge}{\geqslant}

\begin{document}

\title{Computing treedepth in polynomial space and linear fpt time 
	\thanks{This work is a part of projects that have received funding
from the European Research Council (ERC) under the European Union’s Horizon 2020 research and innovation
programme (grant agreements No. 714704 --- {\sc{CUTACOMBS}} and No. 948057 --- {\sc{BOBR}}).}
	}

\date{April 21, 2022}

\author{
	Wojciech Nadara\thanks{Institute of Informatics, University of Warsaw, Poland (\texttt{w.nadara@mimuw.edu.pl})} 
	\and Michał Pilipczuk\thanks{Institute of Informatics, University of Warsaw, Poland
	(\texttt{michal.pilipczuk@mimuw.edu.pl})} \and
	Marcin Smulewicz\thanks{Institute of Informatics, University of Warsaw, Poland (\texttt{m.smulewicz@mimuw.edu.pl})}
}

\maketitle

\begin{textblock}{20}(0, 13.0)
	\includegraphics[width=40px]{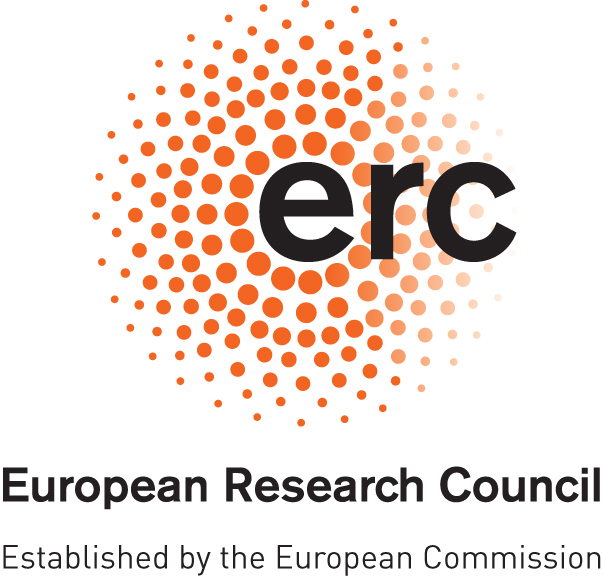}%
\end{textblock}
\begin{textblock}{20}(-0.25, 13.4)
	\includegraphics[width=60px]{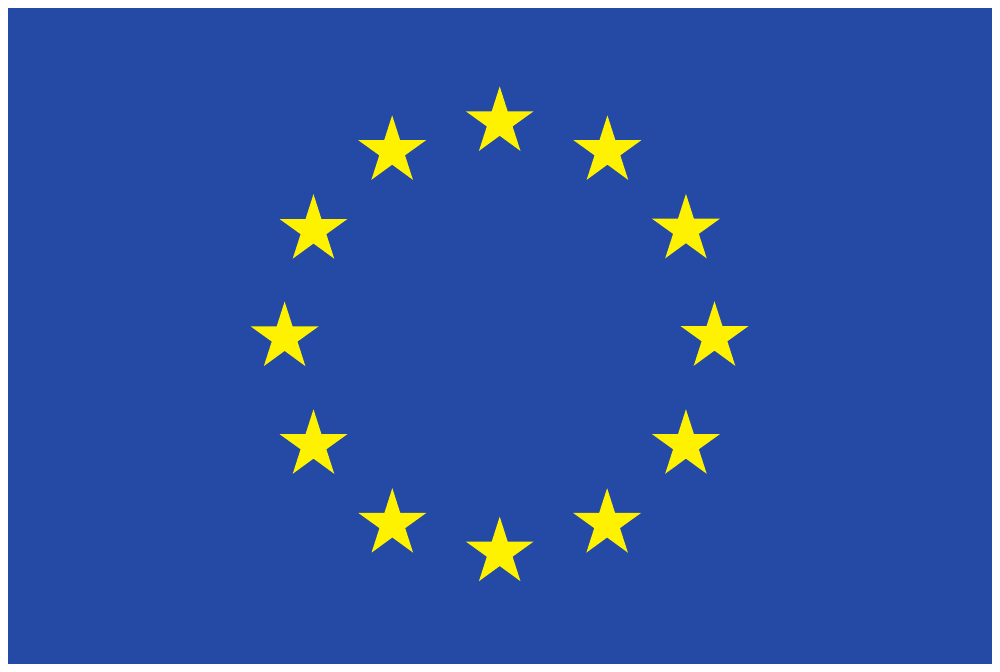}%
\end{textblock}

\begin{abstract}
The {\em{treedepth}} of a graph $G$ is the least possible depth of an {\em{elimination forest}} of $G$: a rooted forest on the same vertex set where every pair of vertices adjacent in $G$ is bound by the ancestor/descendant relation.
We propose an algorithm that given a graph $G$ and an integer $d$, either finds an elimination forest of $G$ of depth at most $d$ or concludes that no such forest exists; thus the algorithm decides whether the treedepth of $G$ is at most~$d$. The running time is $2^{\Oh(d^2)}\cdot n^{\Oh(1)}$ and the space usage is polynomial in $n$. Further, by allowing randomization, the time and space complexities can be improved to $2^{\Oh(d^2)}\cdot n$ and $d^{\Oh(1)}\cdot n$, respectively. This improves upon the algorithm of Reidl et al.~[ICALP 2014], which also has time complexity $2^{\Oh(d^2)}\cdot n$, but uses exponential~space.

\end{abstract}

\newpage

\section{Introduction}
An {\em{elimination forest}} of a graph $G$ is a rooted forest $F$ whose vertex set is the same as that of $G$, where for every edge $uv$ of $G$, either $u$ is an ancestor of $v$ in $F$ or vice versa. The {\em{treedepth}} of $G$ is the least possible depth of an elimination forest of $G$.
Compared to the better-known parameter {\em{treewidth}}, treedepth measures the depth of a tree-like decomposition of a graph, instead of width. The two parameters are related: if by $\td(G)$ and $\tw(G)$ we denote the treedepth and the treewidth of an $n$-vertex graph $G$, then it always holds that $\tw(G)\leq \td(G)\leq \tw(G)\cdot \log_2 n$. However, the two notions are qualitatively different: for instance, a path on $t$ vertices has treewidth $1$ and treedepth $\Theta(\log t)$.  

Treedepth appears prominently in structural graph theory, especially in the theory of sparse graphs of Ne\v{s}et\v{r}il and Ossona de Mendez. There, it serves as a basic building block for fundamental decompositions of sparse graphs --- {\em{low treedepth colorings}} --- which can be used for multiple algorithmic purposes, including designing algorithms for {\sc{Subgraph Isomorphism}} and model-checking First-Order logic. See~\cite[Chapters~6 and~7]{sparsity} for an introduction and~\cite{GajarskyKNMPST20,GroheK09,NesetrilM08a,NesetrilM15,NesetrilM16,OS17,PilipczukST18,PilipczukS21} for examples of applications.

In this work we are interested in using treedepth as a parameter for the design of fixed-parameter (fpt) algorithms. Clearly, every dynamic programming algorithm working on a tree decomposition of a graph can be adjusted to work also on an elimination forest, just because an elimination forest of depth $d$ can be easily transformed into a tree decomposition of width $d-1$. However, it has been observed in~\cite{FurerY14,PiWrochna,HegerfeldK20,NederlofPSW20,PilipczukS21} that for multiple basic problems, 
one can design fpt algorithms working on elimination forests of bounded depth that have polynomial space complexity without sacrificing on the time complexity. These include the following: (In all results below, $n$ is the vertex count and $d$ is the depth of the given elimination forest.)
\begin{itemize}[nosep]
 \item A $3^d\cdot n^{\Oh(1)}$-time $\Oh(d+\log n)$-space algorithm for {\sc{3-Coloring}}~\cite{PiWrochna}.
 \item A $2^d\cdot n^{\Oh(1)}$-time $n^{\Oh(1)}$-space algorithm for counting perfect matchings~\cite{FurerY14}.
 \item A $3^d\cdot n^{\Oh(1)}$-time $n^{\Oh(1)}$-space algorithm for {\sc{Dominating Set}}~\cite{FurerY14,PiWrochna}.
 \item A $d^{|V(H)|}\cdot n^{\Oh(1)}$-time $n^{\Oh(1)}$-space algorithm for {\sc{Subgraph Isomorphism}}~\cite{PilipczukS21}.\\ (Here, $H$ is the sought pattern graph.)
 \item A $3^d\cdot n^{\Oh(1)}$-time $n^{\Oh(1)}$-space algorithm for {\sc{Connected Vertex Cover}}~\cite{HegerfeldK20}.
 \item A $5^d\cdot n^{\Oh(1)}$-time $n^{\Oh(1)}$-space algorithm for {\sc{Hamiltonian Cycle}}~\cite{NederlofPSW20}.
\end{itemize}
We note that the approach used in~\cite{HegerfeldK20,NederlofPSW20} to obtain the last two results applies also to several other problems with connectivity constraints. However, as these algorithms are based on the Cut\&Count technique~\cite{CyganNPPRW11}, they are randomized and no derandomization preserving the polynomial space complexity is known. An in-depth complexity-theoretical analysis of the time-space tradeoffs for algorithms working on different graph decompositions can be found~in~\cite{PiWrochna}.

In the algorithms mentioned above one assumes that the input graph is supplied with an elimination depth of depth at most $d$. Therefore, it is imperative to design algorithms that given the graph alone, computes, possibly approximately, such an elimination forest. Compared to the setting of treewidth and tree decompositions, where multiple approaches have been proposed over years (see e.g.~\cite{Tw5Apx,Korhonen21} for an overview), so far there is only a handful of algorithms to compute the treedepth exactly or approximately.
\begin{itemize}[nosep]
 \item It is well-known (see e.g.~\cite[Section~6.2]{sparsity}) that just running depth-first search and outputing the forest of recursive calls gives an elimination forest of depth at most $2^{\td(G)}$. So this gives a very simple linear-time approximation algorithm, but with the approximation factor exponential in the optimum.
 \item Czerwi\'nski et al.~\cite{TDESA} gave a polynomial-time algorithm that outputs an elimination forest of depth at most $\Oh(\td(G)\tw(G)\log^{3/2} \tw(G))$, which is thus an $\Oh(\tw(G)\log^{3/2} \tw(G))$-approximation algorithm. Recall here that $\tw(G)\leq \td(G)$.
 \item Reidl et al.~\cite{Reidl} gave an exact fpt algorithm that in time $2^{\Oh(d^2)}\cdot n$ either constructs an elimination forest of depth at most $d$, or concludes that the treedepth is larger than $d$. 
\end{itemize}
In particular, obtaining a constant-factor approximation for treedepth running in time $2^{\Oh(\td(G))}\cdot n^{\Oh(1)}$ is a well-known open problem, see e.g.~\cite{TDESA}. We note that implementation of practical fpt algorithms for computing treedepth was the topic of the 2020 Parameterized Algorithms and Computational Experiments (PACE) Challenge~\cite{KowalikMNPSW20}.

\paragraph*{Our contribution.}
The exact algorithm of Reidl et al.~\cite{Reidl} uses not only exponential time (in the treedepth), but also exponential space. This would make it a space bottleneck when applied in combination with any of the polynomial-space algorithms developed in~\cite{FurerY14,PiWrochna,HegerfeldK20,NederlofPSW20,PilipczukS21}. In this work we bridge this issue by proving the following result.

\begin{theorem}\label{thm:main}
 There is an algorithm that given an $n$-vertex graph $G$ and an integer~$d$, either constructs an elimination forest of $G$ of depth at most $d$, or concludes that the treedepth of $G$ is larger than $d$. The algorithm runs in $2^{\Oh(d^2)}\cdot n^{\Oh(1)}$ time and uses $n^{\Oh(1)}$ space. 
 
 The space and time complexities can be improved to $d^{\Oh(1)}\cdot n$ and expected $2^{\Oh(d^2)}\cdot n$, respectively, at the cost of allowing randomization: the algorithm may return a false negative with probability at most $\frac{1}{c\cdot n^c}$, where $c$ is any constant fixed a~priori; there are no false positives. 
\end{theorem}

Thus, the randomized variant of the algorithm of Theorem~\ref{thm:main} has the same time complexity as the algorithm of Reidl et al.~\cite{Reidl}, but uses polynomial space. However, the algorithm of Reidl et al.~\cite{Reidl} is deterministic, contrary to ours. Note that apart from possible false negatives, the bound on the running time is only in expectation and not worst-case (in other words, our algorithm is both Monte Carlo and Las Vegas). However, one can turn this into a worst-case bound at the cost of increasing the probability of false negatives to $1/2$ by forcefully terminating the execution if the algorithm runs for twice as long as expected.

Simultaneously achieving time complexity linear in $n$ and polynomial space complexity is a property that is desired from an algorithm for computing the treedepth of a graph. While many of the polynomial-space  fpt algorithms working on elimination forests do not have time complexity linear in $n$ due to the usage of various algebraic techniques, the simplest ones that exploit only recursion --- like the ones for {\sc{3-Coloring}} or {\sc{Independent Set}} considered in~\cite{PiWrochna} --- can be easily implemented to run in time $2^{\Oh(d)}\cdot n$ and space $d^{\Oh(1)}\cdot n$. Thus, the randomized variant of the algorithm of Theorem~\ref{thm:main} would neither be a bottleneck from the point of view of space complexity nor from the point of view of the dependency of the running time on $n$. Admittedly, the parametric factor in the runtime of our algorithm is $2^{\Oh(d^2)}$, as compared to $2^{\Oh(d)}$ in most of the aforementioned polynomial-space fpt algorithms working on elimination forests; this brings us back to the open problem about constant-factor approximation for treedepth running in time $2^{\Oh(\td(G))}\cdot n^{\Oh(1)}$ raised in~\cite{TDESA}.

Let us briefly discuss the techniques behind the proof of Theorem~\ref{thm:main}. The algorithm of Reidl et al.~\cite{Reidl} starts by approximating the treewidth of the graph (which is upper bounded by the treedepth) and tries to constructs an elimination forest of depth at most $d$ by bottom-up dynamic programming on the obtained tree decomposition. By applying the iterative compression technique, we may instead assume that we are supplied with an elimination forest of depth at most $d+1$, and the task is to construct one of depth at most $d$. 

Applying now the approach of Reidl et al.~\cite{Reidl} directly (that is, after a suitable adjustment from the setting of tree decompositions to the setting of elimination forests) would not give an algorithm with polynomial space complexity. The reason is that their dynamic programming procedure is quite involved and in particular keeps track of certain disjointness conditions; this is a feature that is notoriously difficult to achieve using only polynomial space. Therefore, we resort to the technique of {\em{inclusion-exclusion branching}}, used in previous polynomial-space algorithms working on elimination forests; see~\cite{FurerY14,PiWrochna} for basic applications of this approach. In a nutshell, the idea is to count more general objects where the disjointness contraints are relaxed, and to use inclusion-exclusion at each step of the computation to make sure that objects not satisfying the constraints eventually cancel out. We note that while the application of inclusion-exclusion branching was rather simple in~\cite{FurerY14,PiWrochna}, in our case it poses a considerable technical challenge. In particular, along the way we do not count single values, but rather polynomials with one formal variable that keeps track of how much the disjointness constraints are violated. In the exposition layer, our application of inclusion-exclusion branching mostly follows the algorithm for {\sc{Dominating Set}} of Pilipczuk and Wrochna~\cite{PiWrochna}.

In this way, we can count the number of elimination forests\footnote{Formally, we count only elimination forests satisfying some basic connectivity property, which we call {\em{sensibility}}.} of depth at most $d$ in time $2^{\Oh(d^2)}\cdot n^{\Oh(1)}$ and using polynomial space. So in particular, we can decide whether there exists at least one such elimination forest. Such a decision algorithm can be quite easily turned into a construction algorithm using self-reducibility of the problem. This establishes the first part of Theorem~\ref{thm:main}.

As for the second part --- the randomized linear-time fpt algorithm using polynomial space --- there are several obstacles that need to be overcome. First, there is a multiplicative factor $n$ in the running time coming from the iterative compression scheme. We mitigate this issue by replacing iterative compression with the recursive contraction scheme used by Bodlaender in his linear-time fpt algorithm to compute the treewidth of a graph~\cite{Bodlaender}. Second, when using self-reducibility, we may apply the decision procedure $n$ times, each taking at least linear time. This is replaced by an approach based on color coding, whose correctness relies on the fact that in a connected graph of treedepth at most $d$ there are at most $d^{\Oh(d)}$ different feasible candidates for the root of an optimum-depth elimination tree~\cite{Obstructions2}. Finally, in the counting procedure we may operate on numbers of bitsize as large as polynomial in $n$. This is resolved by hashing them modulo a random prime of magnitude $\Theta(\log n)$, so that we may assume that arithmetic operations take unit time.

We remark that it is relatively rare that a polynomial-space algorithm based on algebraic techniques can be also implemented so that it runs in time linear in the input size. Therefore, we find it interesting and somewhat surprising that this can be achieved for the problem of computing the treedepth of a graph, which combinatorially is rather involved.

\paragraph*{Organization.} After brief preliminaries in Section~\ref{prelims}, in Section~\ref{main} we prove the first part of Theorem~\ref{thm:main}: we give a deterministic algorithm that runs in time $2^{\Oh(d^2)}\cdot n^{\Oh(1)}$ time and uses polynomial space. Then, in Section~\ref{linear} we improve the time and space complexities to $2^{\Oh(d^2)}n$ and $d^{\Oh(1)}n$ respectively, at the cost of introducing randomization.


\section{Preliminaries}\label{prelims}
\newcommand{\Anc}{\mathsf{Anc}}

\paragraph*{Standard notation.}
All graphs in this paper are finite, undirected, and simple (i.e. with no loops on vertices or multiple edges with the same endpoints). For a graph $G$ and a vertex subset $A\subseteq V(G)$, by $N_G[A]$ we denote the {\em{closed neighborhood}} of $A$: the set consisting of all vertices that are in $A$ or have a neighbor in $A$.

For a function $f \colon A \to B$ and a subset of the domain $X\subseteq A$, by $f(X)$ we denote the image of $f$ on $X$. The image of $f$ is denoted $\im(f)=f(A)$. For an element $e$ outside of the domain and a value $\alpha$, by $f[e\to \alpha]$ we denote the extension of $f$ obtained by additionally mapping $e$ to $\alpha$.

We denote the set $\{1, 2, \ldots, k\}$ as $[k]$. We assume the standard word RAM model of computation with words of length $\log n$, where $n$ is the vertex count of the input graph.

\paragraph*{(Elimination) forests and treedepth.}
Consider a rooted forest $F$. By $\Anc_F$ we denote the ancestor/descendant relation in $F$: for $u,v\in V(F)$, $\Anc_F(u, v)$ holds if and only if $u$ is an ancestor of $v$ or $v$ is an ancestor of $u$ in $F$. We assume that a vertex is an ancestor of itself, so in particular $\Anc_F(u, u)$ is always true. We also use the following notation.
For $u \in V(F)$, by $\tail_F[u]$ we denote the set of vertices all ancestors of $u$ (including $u$) and by $\tree_F[u]$ we denote the set of all descendants of $u$, including $u$. 
Further, let $\tail_F(u) = \tail_F[u]\setminus \{u\}$, $\tree_F(u) = \tree_F[u] \setminus \{u\}$, and $\comp_F[u]=\tail_F[u] \cup \tree_F[u]$. Note that $v \in \comp_F[u]$ if and only if $\Anc_F(u, v)$ holds. By $\chld_F(u)$ we denote the set of children of $u$ in $F$, and by $\depth_F(u)$ we denote the depth of $u$ in $F$, that is, $\depth_F(u)=|\tail_F[u]|$ (in particular, roots have depth one). The {\em{depth}} of a rooted forest $F$ is the maximum $\depth_F$ among its vertices. For a set of vertices $A\subseteq F$, by $\cl_F(A)=\bigcup_{u\in A} \tail_F[u]$ we denote the ancestor closure of $A$. A {\em{prefix}} of a rooted forest $F$ is a rooted forest induced by some ancestor-closed set $A\subseteq V(F)$; that is, it is the forest on $A$ with the parent-child relation inherited from $F$.

In this paper we are mostly interested in the notion of an elimination forest and of the treedepth of a graph.

\begin{definition}
	An {\em{elimination forest}} of a graph $G$ is a rooted forest $F$ on the same set of vertices as $G$ such that for every edge $uv \in E(G)$, we have that $\Anc_F(u, v)$ holds. 	The {\em{treedepth}} of a graph $G$ is the least possible depth of an elimination forest of $G$.

\end{definition}

Note that an elimination forest of a connected graph must be connected as well, so in this case we may speak about an {\em{elimination tree}}. Sometimes, instead of identifying $V(G)$ and $V(F)$, we treat them as disjoint sets and additionally provide a bijective mapping $\phi \colon V(G) \to V(F)$ such that $uv \in E(G)$ entails $\Anc_F(\phi(u), \phi(v))$. In such case we consider the pair $(F, \phi)$ to be an elimination forest of $G$. This will be always clear from the context. More generally, for $B\subseteq V(G)$ and a rooted forest $F$, we shall say that a mapping $\phi\colon B\to V(F)$ {\em{respects edges}} if $uv\in E(G)$ entails $\Anc_F(u,v)$ for all $u,v\in B$. In this notation, $(F,\phi)$ is an elimination forest of $G$ if and only if $\phi$ is a bijection from $V(G)$ to $V(F)$ that respects edges on $V(G)$.

\section{Deterministic fpt algorithm} \label{main}
In this section we prove the first part of Theorem~\ref{thm:main}: we give a deterministic polynomial-space algorithm with running time $2^{\Oh(d^2)}\cdot n^{\Oh(1)}$ that for a given $n$-vertex graph $G$, either outputs an elimination forest of $G$ of depth at most $d$ or concludes that no such forest exists. 
The most complex part of the algorithm will be procedure  $\LICZ$, which, roughly speaking, counts the number of different elimination trees of a connected graph $G$ of depth at most $d$. 
We describe $\LICZ$ first, and then we utilize it to achieve the main result.

%
%

\subsection{Description of $\LICZ$}

As mentioned above, procedure $\LICZ$ counts the number of different elimination trees of $G$ of depth at most $d$. However, we will not count all of them, but only such that are in some sense minimal; a precise formulation will follow later. We remark that this part is inspired by the $3^d\cdot n^{\Oh(1)}$-time polynomial space algorithm of Pilipczuk and Wrochna~\cite{PiWrochna} for counting dominating sets in a graph of bounded treedepth. This algorithm exploits the same underlying trick --- sometimes dubbed ``inclusion-exclusion branching'' --- but the application here is technically more involved than in~\cite{PiWrochna}.

%


Before describing $\LICZ$, let us carefully define objects that we are going to count.
We start by recalling the following standard fact about the existence of elimination forests with basic connectivity properties. 

\begin{lemma} \label{ez}
	Let $H$ be a graph and let $R$ be an elimination forest of $H$. Then there exists an elimination forest $R'$ of $H$ such that
	\begin{itemize}[nosep]
		\item for every vertex $u$ of $H$, we have $\depth_{R'}(u)\leq \depth_R(u)$; and
		\item whenever vertices $u,v\in V(H)$ belong to the same connected component of $R'$, they also belong to the same connected component of $H$.
	\end{itemize}
\end{lemma}
\begin{proof}
    For every connected component $C$ of $H$, let $R_C$ be the rooted tree on vertex set $V(C)$ with the ancestor relation inherited from $R$: for $u,v\in V(C)$, $u$ is an ancestor of $v$ in $R_C$ if and only if $u$ is an ancestor of $v$ in $R$. Let $R'$ be the disjoint union of trees $R_C$ over all connected components $C$ of $H$. It is straightforward to check that $R'$ constructed in this way is an elimination forest of $H$ that satisfies both the asserted properties.
\end{proof}

We remark that computing $R'$ can be easily done in linear time by using depth-first search from the root of each elimination tree in $R$. This procedure will be used many times throughout the algorithm when justifying the usual assumption that our current graph is connected. (Disconnected graphs will often naturally appear when recursing after performing some deletions in the original graph.) 

The following lemma can be proved using a very similar, though a bit more involved reasoning. Recall that we work with a fixed connected graph $G$ and its elimination tree~$T$. 

\begin{lemma} \label{sensible}
    Let $G$ be a connected graph of treedepth at most $d$ and $T$ be an elimination tree of $G$ (possibly of depth larger than $d$). Then there exists an elimination tree $R$ of $G$ of depth at most $d$ that satisfies the following property: for every $u\in V(G)$ and $v_1, v_2 \in \chld_T(u)$, $v_1 \neq v_2$, we have
	\begin{equation}\label{eq:bobr}\cl_R(\comp_T[v_1]) \cap \cl_{R}(\comp_T[v_2]) = \cl_{R}(\tail_T[u]).
	\end{equation}
\end{lemma}
\begin{proof}
    Let $R$ be an elimination tree of $G$ of depth at most $d$ that minimizes $\sum_{u\in V(G)}\depth_R(u)$. We claim that $R$ satisfies the required property. Assume otherwise: there are $u\in V(G)$ and distinct $v_1,v_2\in \chld_T(u)$ such that~\eqref{eq:bobr} does not hold. Since the $\cl_R(\cdot)$ operator is monotone under taking subsets, we have
    $$\cl_R(\comp_T[v_1]) \cap \cl_{R}(\comp_T[v_2]) \supseteq \cl_{R}(\tail_T[u])\cap \cl_{R}(\tail_T[u])= \cl_{R}(\tail_T[u]).$$
    So there is $r\in V(G)$ such that $r\notin  \cl_R(\tail_T[u])$, but $r \in \cl_R(\comp_T[v_1])$ and $r \in \cl_R(\comp_T[v_2])$.

    Note that since $r \in \cl_R(\comp_T[v_1]) \setminus \cl_R(\tail_T[u])$, we have $r\in \cl_R(\tree_T[v_1])$.
    So there exists $u_1 \in \tree_T[v_1]$ such that $u_1 \in \tree_R[r]$. Similarly, there exists $u_2 \in \tree_T[v_2]$ such that $u_2 \in \tree_R[r]$. As $r \notin \cl_R(\tail_T[u])$, we have $\tail_T[u] \cap \tree_R[r] = \emptyset$.
    
    Observe that since $T$ is an elimination tree of $G$, we have $N_G[\tree_T[v_1]] \subseteq \comp_T[v_1]$. This implies that $N_{G \setminus \tail_T[u]}[\tree_T[v_1]] \subseteq \tree_T[v_1]$; analogously $N_{G \setminus \tail_T[u]}[\tree_T[v_2]] \subseteq \tree_T[v_2]$.
    So in $G \setminus \tail_T[u]$, each of the sets $\tree_T[v_1]$ and $\tree_T[v_2]$ is the union of vertex sets of a collection of connected components. Note that $\tree_T[v_1]\cap \tree_T[v_2]=\emptyset$.
    
    Consider now the graph $H=G[\tree_R[r]]$. As $\tail_T[u] \cap \tree_R[r] = \emptyset$, $H$ is an induced subgraph of $G \setminus \tail_T[u]$. Further, $H$ intersects both $\tree_T[v_1]$ and $\tree_T[v_2]$, namely $u_1\in V(H)\cap \tree_T[v_1]$ and $u_2\in V(H)\cap \tree_T[v_2]$. Then the conclusion of the previous paragraph implies that $H$ is disconnected.
    
    Let $R_H$ be an elimination forest of $H$ obtained by applying Lemma~\ref{ez} to $H$ and its elimination tree inherited from $R$ and rooted at $r$. Let $R'$ be the elimination tree of $G$ obtained from $R$ by first removing all vertices of $V(H)=\tree_R[r]$, and then reintroducing them again by adding forest $R_H$ and making all roots of $R_H$ into children of the parent of $r$ in $R$. Note that $r$ is not the root of $R$, since $G$ is connected.
    It is straightforward to check that $R'$ is still an elimination tree of $G$, and from Lemma~\ref{ez} it follows that $\depth_{R'}(w)\leq \depth_R(w)$ for each $w\in V(G)$. However, since $R[\tree_R[r]]$ has only one root --- $r$ --- while $R_H$ has at least two roots --- due to $H$ being disconnected --- it follows that $\depth_{R'}(w)<\depth_R(w)$ for at least one $w\in V(H)$. So $R'$ is an elimination tree of $G$ of depth at most $d$ in which the sum of depths of vertices is strictly smaller than in $R$. This is a contradiction with the choice of $R$.
\end{proof}

An elimination tree $R$ of a graph $G$ satisfying the conclusion of Lemma~\ref{sensible} (that is, the depth of $R$ is at most $d$ and for all $u\in V(G)$ and distinct $v_1,v_2\in \chld_T(u)$ we have~\eqref{eq:bobr}) will be called \textit{sensible} with respect to $T$. In our search for elimination trees of low depth, we will restrict attention only to trees that are sensible with respect to some fixed elimination tree $T$. Then Lemma~\ref{sensible} justifies that we may do this without losing all solutions.


With all ingredients introduced, we may finally precisely state the goal of this section.

\begin{lemma}\label{count}
	There exists an algorithm $\LICZ(G, T, d)$ that, given a connected graph $G$ on $n$ vertices, an elimination tree $T$ of depth $k$, and an integer $d$, runs in time $2^{\Oh(dk)}\cdot n^{\Oh(1)}$, uses $n^{\Oh(1)}$ space, and outputs the number of different elimination trees of $G$ of depth at most $d$ that are sensible with respect to $T$.
\end{lemma}

Note here that the input to $\LICZ$ consists not only of $G$ and $d$, but also of an auxiliary elimination tree $T$ of $G$. The depth $k$ of $T$ may be, and typically will be, larger than $d$. Also, we assume that an elimination tree is represented solely by its vertex set and the ancestor relation. In particular, permuting children of a vertex yields the {\em{same}} elimination tree, which should be counted as the same object by procedure $\LICZ$.

The remainder of this section is devoted to the proof of Lemma~\ref{count}. We first need to introduce some definition.

Let us arbitrarily enumerate the vertices of $G$ as $v_1, v_2, \ldots, v_n$ in a top-down manner in~$T$. That is, whenever $v_i$ is an ancestor of $v_j$, we have $i\leq j$. Consider another rooted tree $R$ and a mapping $\phi \colon V(T) \to V(R)$. For a vertex $u$ of $T$, we call a vertex $v_i \in \tree_T(u)$ a \textit{proper surplus image} (for $u$ and $(R,\phi)$) if at least one of the following conditions holds:
\begin{itemize}[nosep]
 \item 
$\phi(v_i) \in \cl_R(\phi(\tail_T[u]))$, or
\item there exists $j$ such that $j<i$, $v_j \in \tree_T(u)$, and $\phi(v_j) = \phi(v_i)$.
\end{itemize} 
We define \textit{non-proper surplus images} analogously, but using sets $\tail_T(u)$ and $\tree_T[u]$ instead of $\tail_T[u]$ and $\tree_T(u)$, respectively.
%


We will work in the ring of polynomials $\Z[x]$, where $x$ is a formal variable.
By an abuse of notation, we equip this ring with an operation of division by $x$ defined through equations:
\begin{gather*}
\frac{x^i}{x} =
\begin{cases}
x^{i-1} & \text{if } i \ge 1,\\
0 & \text{if } i=0
\end{cases}\\[0.3cm]
\frac{\alpha A + \beta B}{x}=\alpha\cdot \frac{A}{x} + \beta\cdot \frac{B}{x}\qquad \textrm{for all}\qquad A, B \in \Z[x]\textrm{ and }\alpha,\beta\in \Z.
\end{gather*}
Formally speaking, division by $x$ is just the unique function from $\Z[x]$ to $\Z[x]$ satisfying the two properties above.

Even though our final goal is to count the number of elimination trees, along the way we are going to count more general objects, called \textit{generalized elimination trees}. A generalized elimination tree of a graph $H$ is a rooted tree $R$ along with a mapping $\phi \colon V(H) \to V(R)$ such that $\phi$ respects edges. Note that in particular, it may be the case that $\im(\phi) \subsetneq V(R)$ or that $\phi(u) = \phi(v)$ for some $u, v\in V(H)$. Clearly, a generalized elimination tree is an elimination tree in the usual sense if and only if $\phi$ is a bijection between $V(H)$ and $V(R)$. We shall call two generalized elimination trees $(R,\phi)$ and $(R',\phi')$ {\em{isomorphic}} if there is an isomorphism of rooted trees $\psi$ mapping $R$ to $R'$ such that $\phi'=\psi\circ \phi$. 

A generalized elimination tree $(R,\phi)$ of an induced subgraph $H$ of $G$ is {\em{sensible}} for $T$ if for every $u\in V(H)$ and distinct $v_1,v_2\in \chld_T(u)\cap V(H)$, we have $\cl_R(\phi(\comp_T[v_1]))\cap \cl_R(\phi(\comp_T[v_2]))=\cl_R(\phi(\tail_T[u]))$. Thus, this notion projects to sensibility of (standard) elimination trees when $H=G$ and $(R,\phi)$ is an elimination tree of $G$. Generalized elimination trees of induced subgraphs of $G$ that are sensible for $T$ shall be called \textit{monsters}.


For a rooted tree $K$, a mapping $\phi$ with co-domain $V(K)$ is called a {\em{cover}} of $K$ if $\cl_K(\im(\phi))=V(K)$, or equivalently, every leaf of $K$ is in the image of $\phi$.
For a vertex $u \in V(G)$, rooted tree $K$ of depth at most $d$,
a subset of vertices $A \subseteq V(K)$ that contains all leaves of $K$, and a mapping $\phi \colon \tail_T(u) \to A$ that is a cover of $K$,
we define $$f(u, K, \phi, A)=\sum_{i=0}^{n} a_i x^i \in \mathbb{Z}[x],$$ where $a_i$ is the number of non-isomorphic monsters $(R, \ol{\phi})$ such that:
\begin{itemize}[nosep]
 \item $(R, \ol{\phi})$ is a generalized elimination tree of $G[\comp_T[u]]$ of depth at most $d$;
 \item $K$ is a prefix of $R$;
 \item $\ol{\phi}$ is an extension of $\phi$ satisfying
 $$V(R) \setminus V(K) \subseteq \im(\ol{\phi}) \subseteq (V(R) \setminus V(K)) \cup A;\qquad\textrm{and}$$
 \item in $\tree_T[u]$ there are exactly $i$ non-proper surplus images for $u$ and $(R,\ol{\phi})$.
\end{itemize}
Note that since $\phi$ is assumed to be a cover of $K$, and by the second and third condition, the last condition can be rephrased as follows: $$i = |\tree_T[u]| - |V(R) \setminus V(K)|.$$

We define polynomial $g(u, K, \phi, L)$ analogously, but using $\tail_T[u]$, $\tree_T(u)$, and proper surplus images, instead of $\tail_T(u)$, $\tree_T[u]$ and non-proper surplus images. That in $\tree_T(u)$ there are $i$ proper surplus images is then equivalent to $i = |\tree_T(u)| - |V(R) \setminus V(K)|$.


Our goal now is to compute the polynomials $f(\cdot,\cdot,\cdot,\cdot)$ and $g(\cdot,\cdot,\cdot,\cdot)$ recursively over the elimination tree $T$.
It can be easily seen that if $\chld_T(u) = \emptyset$ then
\begin{equation}\label{eq:leaf}
g(u, K, \phi, A) = 
\begin{cases}
1& \text{if } \phi \text{ respects edges}, \\
0 & \text{otherwise.}
\end{cases}
\end{equation}
Indeed, $(R,\ol{\phi})=(K,\phi)$ is the only possible pair that can satisfy the last three conditions, and it is a sensible generalized elimination tree of $G[\comp_T[u]]$ if and only if $\phi$ respects edges.

First, we show how to compute polynomials $g(u,\cdot,\cdot,\cdot)$ based on the knowledge of polynomiasl $f(v,\cdot,\cdot,\cdot)$ for children $v$ of $u$.

\begin{lemma} \label{prod}
If $\chld_T(u) \neq \emptyset$, then for all relevant $u,K,\phi,A$ we have
$$g(u, K, \phi, A) = \prod_{v \in \chld_T(u)} f(v, K, \phi, A)$$
\end{lemma}
\begin{proof}
	Let $\chld_T(u) = \{v_1, \ldots, v_c\}$ and let $(R_1, \ol{\phi}_1), \ldots, (R_c, \ol{\phi}_c)$ be any monsters such that $(R_i, \ol{\phi}_i)$ is a monster counted in the definition of $f(v_i, K, \phi, A)$. 
	Note that $K$ is a prefix of each $R_i$, and each $\ol{\phi}_i$ is an extension of $\phi$. Therefore, we can construct a monster $(R, \ol{\phi})$ as follows:
	\begin{itemize}[nosep] 
	 \item $R$ is the union of $R_1,\ldots,R_c$ with the vertices of $K$ identified naturally;
	 \item $\ol{\phi}$ is the union of $\ol{\phi}_1,\ldots, \ol{\phi}_c$ (note that values on $K$ match).
	\end{itemize}
	That $(R,\ol{\phi})$ constructed in this manner is sensible for $T$ is easy to verify.
	Moreover, observe that every distinct tuple of monsters $(R_1, \ol{\phi}_1), \ldots, (R_c, \ol{\phi}_c)$ gives rise to a different (non-isomorphic) monster~$(R, \ol{\phi})$.
	
	On the other hand, we argue that every monster $(R, \ol{\phi})$ counted in the definition of $g(u,K,\phi,A)$ can be obtained from some monsters $(R_1, \ol{\phi}_1), \ldots, (R_c, \ol{\phi}_c)$ in the way described above. Indeed, $(R,\ol{\phi})$ is sensible for $T$, hence every subtree of $R-K$ accommodates images under $\ol{\phi}$ of vertices from only one subtree $\tree_T[v_i]$, for some $i\in \{1,\ldots,c\}$. Distributing the subtrees of $R-K$ according to the index $i$ as above naturally gives rise to monsters $(R_1, \ol{\phi}_1), \ldots, (R_c, \ol{\phi}_c)$ that are counted in the definitions of $f(v_1, K, \phi, A),\ldots,f(v_c, K, \phi, A)$, respectively.
	
	Altogether, we have shown that distinct tuples of monsters $(R_1, \ol{\phi}_1), \ldots, (R_c, \ol{\phi}_c)$ contributing to the definitions of $f(v_1, K, \phi, A),\ldots,f(v_c, K, \phi, A)$ are in one-to-one correspondence with monsters $(R,\ol{\phi})$ contributing to the definition of $g(u, K, \phi, A)$. This correspondence preserves the number of surplus vertices in the following sense: if for eah $i\in \{1,\ldots,c\}$, $\tree_T[v_i]$ has $j_i$ non-proper surplus images for $v_i$ and $(R_i, \ol{\phi}_i)$, then $\tree_T(u)$ has $j_1 + \ldots + j_c$ proper surplus images for $u$ and $(R, \ol{\phi})$. This directly implies the postulated equality of polynomials. 
\end{proof}

Let us elaborate on the intuition on what happened in Lemma~\ref{sensible}. Intuitively, we aggregated information about the children of $u$ to the information about $u$ itself. Since in the definitions of monsters we do not insist on the mappings being injective, this aggregation could have been performed by a simple product of polynomials (though, the assumption of sensibility was crucial for arguing the correctness). In a natural dynamic programming, such as the one in~\cite{Reidl}, one would need to ensure injectivity when aggregating information from the children of $u$, which would result in a dynamic programming procedure that would need to keep track of all subsets of $K$ (and thus use exponential space). Thus, relaxing injectivity here allows us to use simple multiplication of polynomials, but obviously we will eventually need to enforce injectivity. The idea is that we enforce surjectivity instead, and make sure that the size of the co-domain matches the size of the domain. In turn, surjectivity is enforced using inclusion-exclusion in the computation of polynomials $f(u,\cdot,\cdot,\cdot)$ based on polynomials $g(u,\cdot,\cdot,\cdot)$, which is the subject of the next lemma.
%
%
\begin{lemma}  \label{kobyla}
	For all relevant $u,K,\phi,A$, we have:
\begin{equation*}
\begin{aligned}
f(u, K, \phi, A) = {} & \sum_{v \in A} x \cdot g(u, K, \phi[u \to v], A) + \\
& \sum_{w \in K} \sum_{p=1}^{d - \depth(w)} \frac{1}{x^{p-1}} \\
&\sum_{B \subseteq \{w_1, \ldots, w_{p-1}\}} (-1)^{p-1-|B|}
g(u, K[w, w_1, \ldots ,w_{p}], \phi[u \to w_p], A \cup B \cup \{w_p\}), 
\end{aligned}
\end{equation*}
where $K[w,w_1,\ldots,w_p]$ denotes the rooted tree obtained from $K$ by adding a path $[w,w_1,\ldots,w_p]$ so that $w$ is the parent of $w_1$ and each $w_i$ is the parent of $w_{i+1}$, for $i\in \{1,\ldots,p-1\}$.
\end{lemma}
\begin{proof}
Let $(R,\ol{\phi})$ be a monster counted in the definition of $f(u,K,\phi,A)$. Observe that $u$ is in the domain of $\ol{\phi}$, but not in the domain of $\phi$. The intuition is that extending $\phi$ by mapping $u$ to $\ol{\phi}(u)$ yields an object that is indirectly taken into account in the polynomials $g(u,\cdot,\cdot,\cdot)$, but we need to be careful that we express the contribution of $(R,\ol{\phi})$ to $f(u,K,\phi,A)$ as a combination of contributions of different monsters to different polynomials $g(u,\cdot,\cdot,\cdot)$. Let~$v=\ol{\phi}(u)$.

Consider first the case when $v\in V(K)$. Note that then we necessarily have $v\in A$. Then $(R,\ol{\phi})$ is a monster that is counted in the definition of $g(u,K,\phi[u\to v],A)$. Observe that $u$ is a non-proper surplus image for $u$ and $(R,\ol{\phi})$, but it is not a proper surplus image for $u$ and $(R,\ol{\phi})$, hence the number of proper surplus images for $u$ and $(R,\ol{\phi})$ is exactly one larger than the number of non-proper surplus images for $u$ and $(R,\phi)$. Also, every monster counted in the definition of $g(u,K,\phi[u\to v],A)$ contributes to $f(u,K,\phi,A)$ as above. This justifies the summand $\sum_{v\in A} x \cdot g(u, K, \phi[u \to v], A)$ in the formula.

Consider now the case when $v\in V(R)\setminus V(K)$. 
We need to consider various cases on how $\cl_R(\ol{\phi}(\tail[u]))$ differs from $\cl_R(\ol{\phi}(\tail(u)))$. The former can be described as the latter with a path attached, connecting $v$ with the least ancestor $w$ of $v$ that belongs to $\cl_R(\ol{\phi}(\tail(u)))$.
Let this path be $P=[w,w_1,\ldots,w_p]$, where $w_p=v$, and observe that the length of $P$, call it $p$, satisfies $p+\depth(w)\leq d$. Therefore, if we denote $K'=K[w,w_1,\ldots,w_p]$, then it the case that $(R,\ol{\phi})$ is a monster that is counted in the definition of $g(u,K',\phi[u\to w_p],A\cup \{w_1,\ldots,w_p\})$. The problem is that not every monster counted in the definition of $g(u,K',\phi[u\to w_p],A\cup \{w_1,\ldots,w_p\})$ contributes to $f(u,K,\phi,A)$, because in the definition of the latter we require that $\ol{\phi}$ is surjective onto $V(R)\setminus V(K)$.
%

This issue is mitigated using the inclusion-exclusion principle. We iterate over all subsets $B\subseteq \{w_1,\ldots,w_{p-1}\}$ and take into account the contribution from $g(u,K',\phi[u\to w_p],A\cup B \cup \{w_p\})$ with sign $(-1)^{p-1-|B|}$. In this way, the only monsters that survive in the summation are those corresponding to monsters that are surjective onto $\{w_1,\ldots,w_{p-1}\}$. 

Finally, we need to be careful about properly counting surplus images through the degrees of the formal variable $x$. As argued, the only summands that survive inclusion-exclusion summation are those corresponding to monsters $(R,\ol{\phi})$ where $\{w_1,\ldots,w_{p-1}\}\subseteq \im(\ol{\phi})$; so fix such a monster. For each $j \in \{1,\ldots,p-1\}$ there is the smallest index $s(j)$ such that $v_{s(j)} \in \tree(u)$ and $\ol{\phi}(v_{s(j)}) = w_j$. Then $v_{s(j)}$ is a proper surplus image for $u$ and $(R,\ol{\phi})$, but is not a non-proper surplus image for $u$ and $(R,\ol{\phi})$. It is straightforward to check that all vertices of $ \tree[u]$ except for $v_{s(1)},\ldots,v_{s(p)}$ retain their status: they are a proper surplus image for $u$ and $(R,\ol{\phi})$ if and only if they are a non-proper surplus image for $(R,\ol{\phi})$. Hence, there are exactly $p-1$ more proper surplus images for $u$ and $(R,\ol{\phi})$ than there are non-proper surplus images  for $u$ and $(R,\ol{\phi})$. This justifies dividing the result of the inclusion-exclusion summation by ${x^{p-1}}$ and concludes the proof.
\end{proof}


We need to take an additional care of how to deduce the overall number of elimination trees based on the polynomial $f(\cdot,\cdot,\cdot,\cdot)$ and $g(\cdot,\cdot,\cdot,\cdot)$. Define polynomial
$$h = \sum_{p=1}^{d} \frac{1}{x^{p-1}} \sum_{B \subseteq \{w_1, \ldots, w_{p-1}\}} (-1)^{p-1-|B|}
g(r, [w_1, \ldots, w_{p}], [r \to w_p], B \cup \{w_p\})\in \Z[x],$$
where $r$ is the root of $T$, $[w_1, \ldots, w_{p}]$ is a path on $p$ vertices rooted at $w_1$, and $[r\to w_p]$ denotes the function with domain $\{r\}$ that maps $r$ to $w_p$.

\begin{lemma} \label{init}
	The number of elimination trees of $G$ that are sensible with respect to $T$ and have depth at most $d$ is the term in $h$ standing by $x^0$.
\end{lemma}

\begin{proof}
    By Lemma~\ref{kobyla},
	the formula can be seen as the formula for $f(r, K, \phi, A)$ for empty $K$, $\phi$, and $A$. Therefore, $h$ can be written as $h=\sum_{i=0}^{n} a_i x^i$, where $a_i$ is the number of non-isomorphic sensible generalized elimination trees $(R, \ol{\phi})$ such that $R$ has depth at most $d$, $\ol{\phi}\colon V(G)\to V(R)$ is surjective, and in $G$ there are $i$ non-proper surplus images for $r$ and $(R,\ol{\phi})$. However, since $K$ is empty, the number of surplus images is exactly the number of vertices $v_j\in V(G)$ that are mapped by $\ol{\phi}$ to the same vertex of $R$ as some other vertex of $G$ with a smaller index. Then the assertion that $\ol{\phi}$ is injective is equivalent to the assertion that the number of such surplus images is $0$. It follows that the number of non-isomorphic sensible elimination trees of $G$ of depth at most $d$ is equal to the term in $h$ that stands by $x^0$.
\end{proof}

Having established Lemmas~\ref{prod},~\ref{kobyla} and~\ref{init}, we can conclude the description of  procedure $\LICZ$. By~\ref{init}, the goal is to compute polynomial $h$ and return the coefficient standing by $x^0$. We initiate the computation using the formula for $h$, and then we use two mutually-recursive procedures to compute polynomials $f(\cdot,\cdot,\cdot,\cdot)$ and $g(\cdot,\cdot,\cdot,\cdot)$ using formulas provided by  Lemmas~\ref{prod} and~\ref{kobyla}. The base case of recursion is for a leaf of $T$, where we use formula~\eqref{eq:leaf}.

The correctness of the procedure is established by Lemmas~\ref{prod},~\ref{kobyla} and~\ref{init}.
So it remains to bound its time complexity and memory usage. It is clear that polynomials that we compute will always have degrees at most $n$. Trees $K$ relevant in the computation will never have more than $dk$ vertices, for at every recursive call the tree $K$ can grow by at most $d$ new vertices.

As the next step, we bound the numbers that can be present in the computations.
\begin{lemma} \label{bound}
	Every coefficient of $f(u, K, \phi, A)$ is an integer from the range $[0,(dk \cdot 2^d)^{|\tree_T[u]|}]$ and every coefficient of $g(u, K, \phi, A)$ is an integer from the range $[0,(dk \cdot 2^d)^{|\tree_T(u)|}]$. Hence, all integers present in the computations are at most $(dk2^d)^n$.
\end{lemma}
\begin{proof}
	We prove this by induction on the recursion tree. The base of the induction (that is, calls of $g$ on leaves) is clear. Induction step for $g$ called on a vertex $u$ that is not a leaf is clearly following from the bounds on $f$ called on children of $u$ as $|\tree_T(u)| = \sum_{v \in \chld_T(u)} |\tree_t[v]|$. Induction step for $f$ called on a vertex $u$ follows from the fact that it is a sum of at most $|A| + |K| \cdot (2^0 + 2^1 + \ldots + 2^{d-1}) \le dk + dk \cdot (2^d-1) = dk \cdot 2^d$ calls of $g$ with coefficients from the set $\{-1, 1\}$ on the same vertex and the fact that $|\tree_T[u]| = |\tree_T(u)| + 1$. 
\end{proof}

It follows that all integers  present in the computation have bitsize bounded polynomially in $n$.

As for the memory usage, the run of the algorithm is a recursion of depth bounded by $2k$. The memory used is a stack of at most $2k$ frames for recursive calls of procedures computing polynomials $f(\cdot,\cdot,\cdot,\cdot)$ and $g(\cdot,\cdot,\cdot,\cdot)$ for relevant arguments. Each of these frames requires space polynomial in $n$, hence the total space complexity is polynomial in $n$.

As for the time complexity, each call to a procedure computing a polynomial of the form $f(u,\cdot,\cdot,\cdot)$ makes at most $dk\cdot 2^d$ recursive calls to procedures computing polynomials of the form $g(u,\cdot,\cdot,\cdot)$. In turn, each of these calls makes one call to a procedure computing a polynomial of the form $f(v,\cdot,\cdot,\cdot)$ for each child $v$ of $u$. It follows that the total number of calls to procedures computing polynomials of the form $f(u,\cdot,\cdot,\cdot)$ and $g(u,\cdot,\cdot,\cdot)$ is bounded by $2\cdot (dk\cdot 2^d)^k=2^{\Oh(dk)}$. The internal work needed in each recursive call is bounded by $2^{\Oh(d)}\cdot n^{\Oh(1)}$. As $T$ has $n$ vertices, the total time complexity is $2^{\Oh(dk)}\cdot 2^{\Oh(d)}\cdot n^{\Oh(1)}\cdot n= 2^{\Oh(dk)}\cdot n^{\Oh(1)}$, as claimed.
This concludes the proof of Lemma~\ref{count}.

We note that having designed $\LICZ(G, T, d)$, it is easy to design a similar function $\LICZLAS(G, T, d)$ that does not need an assumption of $G$ being connected and where $T$ is some elimination forest instead of an elimination tree (by using the procedure described after Lemma \ref{ez}).

%
%

\subsection{Utilizing $\LICZ$} \label{odzyskaj}

With the description of $\LICZ$ completed, we can describe how we can utilize it in order to construct a bounded-depth elimination tree of a graph. That is, we prove the first part of Theorem~\ref{thm:main}.

First, we lift $\LICZ$ to a constructive procedure that still requires to be provided an auxiliary elimination tree of the graph.

\begin{lemma} \label{get}
	There is an algorithm $\KONSTRUUJ(G, T, d)$ that, given an $n$-vertex graph $G$, an elimination forest $T$ of $G$ of depth at most $k$, and an integer $d$, runs in time $2^{\Oh(dk)} \cdot n^{\Oh(1)}$, uses $n^{\Oh(1)}$ space, and either correctly concludes that $\td(G)>d$ or returns an elimination forest of $G$ of depth at most $d$.
\end{lemma}

\begin{proof}
    By treating every connected component separately, we may assume that $G$ is connected (see the remark after Lemma \ref{ez}). Thus $T$ is an elimination tree of $G$.

	The first step of $\KONSTRUUJ(G, T, d)$ is calling $\LICZ(G, T, d)$. If this call returns $0$, we terminate $\KONSTRUUJ$ and report that $\td(G)>d$; this is correct by Lemma~\ref{sensible}. Otherwise we are sure that $\td(G)\leq d$, and we need to construct any elimination tree of depth at most $d$. In order to do so, we check, for every vertex $v \in V(G)$, whether $v$ is a feasible  candidate for the root of desired elimination tree. Note that a vertex $v$ can be the root of an elimination tree of $G$ of depth at most $d$ if and only if $\td(G-v)<d$, or equivalently, if an only if 
	the procedure $\LICZLAS(G \setminus v, T-v, d - 1)$ returns a positive value. (Here, by $T-v$ we mean the forest $T$ with $v$ removed and all former children of $v$ made into children of the parent of $v$, or to roots in case $v$ was a root.) As $\td(G)\leq d$, we know that for at least one vertex $v$, this check will return a positive outcome. Then we recursively call $\KONSTRUUJ(G-v, T-v, d - 1)$, thus obtaining an elimination forest $F'$ of $G-v$ of depth at most $d-1$, and we turn it into an elimination tree $F$ of $G$ by adding $v$ as the new root and making it the parent of all the roots of $F'$. As $F$ has depth at most $d$, it can be returned as the result of the procedure.

	That the procedure is correct is clear. As for the time and space complexity, it is easy to see that there will be at most $dn$ calls to the procedure $\LICZ$ in total, because at each level of the recursion there will be at most one invocation of $\LICZ$ per vertex of the original graph. As each of these calls uses $2^{\Oh(dk)}\cdot n^{\Oh(1)}$ time and $n^{\Oh(1)}$ space, the same complexity bounds also follow for $\KONSTRUUJ$.
\end{proof}

It remains to show how to lift the assumption of being provided an auxiliary elimination forest of bounded depth. For this we use the iterative compression technique.

\begin{proof}[Proof of the first part of Theorem~\ref{thm:main}] 
    Arbitrarily enumerate the vertices of $G$ as $v_1, v_2, \ldots, v_n$. For $i\in \{1,\ldots,n\}$, let $G_i=G[\{v_1,\ldots,v_i\}]$ be the graph induced by the first $i$ vertices. For each $i=1,2,\ldots, n$ we will compute $F_i$, an elimination forest of $G_i$ of depth at most $d$. For $i=1$ this is trivial. Assume now that we have already computed $F_i$ and want to compute $F_{i+1}$. We first construct $T_{i+1}$, an elimination tree of $G_{i+1}$, by taking $F_i$, adding $v_{i+1}$, and making $v_{i+1}$ the parent of all the roots of $F_i$. Note that $T_{i+1}$ has depth at most $d+1$. We now call $\KONSTRUUJ(G_{i+1}, T_{i+1}, d)$. If this procedure concludes that $\td(G_{i+1})>d$, then this implies that $\td(G)>d$ as well, and we can terminate the algorithm and provide a negative answer. Otherwise, the procedure returns an elimination forest $F_{i+1}$ of $G_{i+1}$ of depth at most $d$, with which we can proceed. Eventually, the algorithm constructs an elimination forest $F=F_n$ of $G=G_n$ of depth at most $d$.
    
    The algorithm is clearly correct. Since every call to $\KONSTRUUJ$ is supplied with an elimination forest of depth at most $d+1$, and there are at most $n$ calls, the total time complexity is $2^{\Oh(d^2)} \cdot n^{\Oh(1)}$ and the space complexity is $n^{\Oh(1)}$, as desired.
\end{proof}

\section{Randomized linear fpt algorithm} \label{linear}
\newcommand{\Solve}{\mathtt{Solve}}

In this section we prove the second part of Theorem~\ref{thm:main}: we reduce the time and space complexities to linear in $n$ at the cost of relying on randomization. There are three main reasons why the algorithm presented in the previous section does not run in time linear in $n$.
\begin{itemize}[nosep]
 \item First, in procedure $\ODZYSKAJ$, we applied $\LICZ$ $\Oh(dn)$ times. Even if $\LICZ$ runs in time linear in $n$, this gives at least a quadratic time complexity for $\ODZYSKAJ$.
 \item Second, in the iterative compression scheme we add vertices one by one and apply procedure $\ODZYSKAJ$ $n$ times. Again, even if $\ODZYSKAJ$ runs in linear time, this gives at least a quadratic time complexity.
 \item Third, in procedure $\LICZ$ we handle polynomials of degree at most $n$ and with coefficients of bitsize bounded only polynomially in $n$. Algebraic operations on those need time polynomial in $n$.
\end{itemize}
In short, these obstacles are mitigated as follows:
\begin{itemize}[nosep]
 \item We give another implementation of $\ODZYSKAJ$ that applies a modified variant of $\LICZ$ only $d^{\Oh(d)}$ times. In essence, we sample a random coloring of the graph with $d^{\Oh(d)}$ colors, and for every color we apply a modification of $\LICZ$ that is able to pinpoint a candidate for the root of an optimum-depth elimination forest in this color, provided there is exactly one. Since the total number of candidates in a connected graph of treedepth at most $d$ is at most $d^{\Oh(d)}$~\cite{Obstructions2}, this procedure finds a candidate root with high probability.
 \item Iterative compression is replaced by a contraction scheme of Bodlaender~\cite{Bodlaender} that allows us to replace iteration with recursion, where every recursive step reduces the total number of vertices by a constant fraction, rather than peels off just one vertex.
 \item We observe that in $\LICZ$, we may care only about monomials with degrees bounded by $dk$, so the degrees are not a problem. As for coefficients, we hash them modulo a sufficiently large prime. This is another source of randomization.
\end{itemize}
We proceed to formal details.

%

\subsection{Optimizing the running time of $\LICZ$}
%

We deal with monomials of high degree first.

\begin{lemma} \label{clip}
	The output of $\LICZ(G, T, d)$ does not change if we use the quotient ring $\Z[x]/(x^{dk})$ instead of $\Z[x]$.
\end{lemma}
\begin{proof}
	Recall that the final output of $\LICZ$ is the {\em{free term}} of the polynomial $h$, that is, the coefficient standing by $x^0$. The only division by $x$ in the whole algorithm happens in the formula provided by Lemma~\ref{kobyla}, where we divide by $x^{p-1}$, where $p \le d$. On any path of recursive calls in our algorithm, there are at most $k$ calls of this type, hence the summands of form $x^i$ for $i>k(d-1)$ will never have any contribution to the free term in the polynomial returned at the root of the recursion. Therefore, ignoring those summands does not affect the final result of the computation.
\end{proof}

Now, we optimize the cost of arithmetic operations. To this end, we use the standard technique of performing arithmetic operations modulo a random prime.

We start by recalling the following fact \cite[Fact 29]{PiWrochna}, which is based on \cite[Theorem 4]{Primes}.
\begin{fact} \label{primes}
	There is a positive integer $L$ such that for all integers $\ell \ge L$ it holds that the product of primes strictly between $\ell$ and $2\ell$ is larger than $2^\ell$.
\end{fact}

Let $C \ge 1$ be a sufficiently large constant, to be specified later. Let $$A = \max(L, n^5 2^{5Cd^2}).$$ Recall that by Lemma~\ref{bound}, the coefficients that appear during the computation of $\LICZ$ are upper bounded $(dk 2^d)^n$. However, due to modifications that will be explained later when we will speak about the weighted variant of procedure $\LICZ$, we will actually need to perform arithmetics on numbers as large as $(ndk2^d)^n$. As $k$ in our applications is never larger than $2d$ and as $2d \le 2^d$ for positive integers $d$, we have $(ndk2^d)^n \le (n2^{3d})^n$. Positive integers that are at most that large cannot have more than $n$ distinct prime factors in the interval $(A, 2A)$, as $A^n > (n2^{3d})^n$. However, by Fact~\ref{primes}, we know that there are at least $\frac{\log_2{2^A}}{\log_2{2A}} = \frac{A}{\log_2{A} + 1} = \Omega(n^4 2^{4Cd^2})$ primes in this interval. Since each non-zero number $x$ that appears  in the computation has no more than $n$ distinct prime factors in the interval $(A, 2A)$, it means that the probability that a (uniformly sampled) random prime from this interval divides $x$ is at most $n \cdot \Oh\left(\frac{1}{n^4 2^{4Cd^2}}\right) = \Oh\left(\frac{1}{n^3 2^{4Cd^2}}\right)$. 

Consider the procedure $\LICZ$ modified as follows: at the beginning we sample uniformly at random a prime $p\in (A,2A)$ and instead of computing every number explicitly, we work in the ring $\Z_p=\Z/(p)$ and thus only compute the remainders modulo $p$. If the number of elimination trees of $G$ of depth at most $d$ is $0$, then we are sure that this algorithm eventually obtains $0$ as well. However, if this number is nonzero, then this algorithm will obtain $0$ modulo $p$ with probability is at most $\Oh\left(\frac{1}{n^3 2^{4Cd^2}}\right)$. Note that the bitsize of $p$ is $d^{\Oh(1)} + \Oh(\log n)$, hence all arithmetic operations in $\Z_p$ can be performed in $d^{\Oh(1)}$ time in the RAM model. Hence, by working in the ring $\Z_p$ for a random prime $p\in (A,2A)$, we significantly improve the cost of arithmetic operations while sacrificing only a little in terms of the correctness. That is, testing whether the number of elimination trees of $G$ of depth at most $d$ is nonzero may result in a false negative with probability $\Oh\left(\frac{1}{n^3 2^{4Cd^2}}\right)$, so we have a Monte Carlo algorithm with one-sided error. Throughout the remaining part of this article, we are sometimes going to refer to numbers that are present in the computation of $\LICZ$ working in $\Z$ as \textit{true} numbers, as opposed to their corresponding remainders that appear in the computation where $\LICZ$ works in $\Z_m$ for some number $m$.

Let us briefly describe how we sample a random prime from the interval $(A, 2A)$. We repeat the following procedure until we find the first prime: We first uniformly sample a random integer from this interval and then we check whether it is prime using the AKS primality test~\cite{AKS}. As argued before, there are at least $\frac{A}{\log_2{A} + 1}$ primes in this interval, hence the probability of finding a prime when sampling a random number from this interval is at least $\frac{1}{\log_2{A} + 1}$. Therefore, the expected number of trials needed to sample a prime will be at most $\log_2{A} + 1$. The AKS primality test works in time $(\log{A})^{\Oh(1)}$, hence the expected work spent till discovering a prime is $(\log{A})^{\Oh(1)} = (d \log n)^{\Oh(1)} \subseteq d^{\Oh(1)} n$. So this is a Las Vegas algorithm (which obviously can be turned into a Monte Carlo algorithm by stopping it after a certain number of failed trials). Note that we draw only one random prime $p$ at the very beginning of our algorithm, and whenever we want to use a prime, we use this one.

After improving both the degrees of involved polynomials and the cost of arithmetic operations, single call of $\LICZ$ in its current version takes $2^{\Oh(dk)} n$ time. 

\subsection{Faster root recovery}

\newcommand{\imp}[2]{#1^{\langle #2\rangle}}

Having improved the running time of $\LICZ$ to linear, now we are going to improve the running time of $\ODZYSKAJ$ to linear. Recall that $\ODZYSKAJ$ in its current version iterates over all vertices $v \in V(G)$ and checks whether $\td(G - v) \le d-1$ (by calling $\LICZ$ with appropriate parameters) --- such vertices $v$ could be placed as roots of an elimination tree of $G$ of depth at most $d$. Finding any feasible root is the crucial part that needs to be optimized in order to achieve a linear running time for $\ODZYSKAJ$. The key fact we are going to use is that the number of possible roots of optimum-depth elimination forests of a connected graph is bounded in terms of the treedepth~\cite{Obstructions2, Obstructions1}.

We need a definition. 
\begin{definition}
	We say that a graph $G$ is a \emph{minimal obstruction} for treedepth $d$ if $\td(G)>d$, but $\td(G - v) \leq d$ for each $v \in V(G)$.
\end{definition}
 
 Dvořák et al.~\cite{Obstructions1} proved that every minimal obstruction for treedepth $d$ satisfies $|V(G)| \le 2^{2^{d-1}}$. This bound was later on improved by Chen et al.~\cite{Obstructions2} to $d^{\Oh(d)}$. An easy consequence of these facts is the following:

\begin{lemma} \label{roots-num}
	Suppose $G$ is a graph whose treedepth is equal to $d$. Then  there are at most $d^{\Oh(d)}$ vertices $v \in V(G)$ such that $\td(G - v) < d$.
\end{lemma}
\begin{proof}
    Let $G'$ be an inclusion-wise minimal induced subgraph of $G$ satisfying $\td(G')=d$. By minimality, $G'$ is a minimal obstruction for treedepth $d-1$. So by the result of Chen et al.~\cite{Obstructions2}, $|V(G')| \in d^{\Oh(d)}$. Note that for every $v\in V(G)\setminus V(G')$ we have $\td(G-v)=d$, for in such case $G-v$ contains $G'$ as an induced subgraph and $\td(G')=d$. So, the number of vertices $v$ satisfying $\td(G-v)<d$ is bounded by $|V(G')|$, which in turn is bounded by $d^{\Oh(d)}$. 
\end{proof}

Note that any improvement in the upper bound on the sizes of obstructions entails an analogous improvement in the bound of Lemma~\ref{roots-num}. Also observe that supposing $G$ is connected, vertices $v$ satisfying $\td(G-v)<\td(G)$ are exactly those that can be placed as roots of an optimum-depth elimination tree.

As the next step, we are going to modify the procedure $\LICZ(G, T, d)$ by introducing weights. Let $G$ be a connected graph. Enumerate vertices of $G$ as $V(G) = \{v_1, \ldots, v_n\}$ and let $t_i$ be the number of elimination trees of $G$ that are sensible with respect to $T$ and in which $v_i$ is the root. Then the result of $\LICZ(G, T, d)$ can be expressed as $t_1 + t_2 + \ldots + t_n$. However, with a slight modification, we are able to compute $t_1 \mu_1 + t_2 \mu_2 + \ldots + t_n \mu_n$ for any sequence $\mu_1, \mu_2, \ldots, \mu_n\in \Z$. In order to do so, we change the formula from Lemma~\ref{kobyla} to the following:
\begin{equation*}
\begin{aligned}
f(u, K, \phi, A) = {} & \sum_{v \in A} x \cdot g(u, K, \phi[u \to v], A) \cdot \mu(u, v)+ \\
& \sum_{w \in K} \sum_{p=1}^{d - \depth(w)} \frac{1}{x^{p-1}} \\
&\sum_{B \subseteq \{w_1, \ldots, w_{p-1}\}} (-1)^{p-1-|B|}
g(u, K[w, w_1, \ldots ,w_{p}], \phi[u \to w_p], A \cup B \cup \{w_p\}), 
\end{aligned}
\end{equation*}
where 
\begin{equation*}
\mu(v_i, u) = 
\begin{cases}
\mu_i& \text{if } u \text{ is the root of } K, \\
1 & \text{otherwise.}
\end{cases}
\end{equation*}

Similarly, we adjust the formula for the polynomial $h$:

$$h = \sum_{p=1}^{d} \frac{1}{x^{p-1}} \sum_{B \subseteq \{w_1, \ldots, w_{p-1}\}} (-1)^{p-1-|B|}
g(r, [w_1, \ldots, w_{p}], [r \to w_p], B \cup \{w_p\}) \cdot \mu(r, w_p)$$
($w_p$ is the root of the path $[w_1, \ldots, w_p]$ if and only if $p=1$).

Naturally, the definition of $f(\cdot,\cdot,\cdot,\cdot)$ and $g(\cdot,\cdot,\cdot,\cdot)$ change as well. Instead of simply counting monsters in a weighted fashion so that the contribution of every monster to the sum is the product of numbers $\mu_i$ over all $v_i$-s that were mapped onto the root in the monster (the empty product is assumed to be equal to $1$). However, we already know that the contribution of each monster that is not a valid elimination tree cancels out, so only valid elimination trees remain in the final result. For these, exactly one vertex was mapped to the root of the generalized elimination tree, hence the contribution of each such elimination tree is $\mu_i$ instead of $1$, where $v_i$ is the vertex that is mapped to the root. All in all, the final result is indeed equal to $\sum_{i=1}^{n} t_i \mu_i$, as claimed.

Assume wishfully that there is exactly one vertex $v_i \in V(G)$ that could serve as the root of an elimination tree of $G$ of depth $d$; equivalently, $v_i$ is the only vertex such that $\td(G-v_i)<d$. In other words, $t_j$ is nonzero if and only if $i=j$. Note that in such case we have $i = \frac{\sum_{j=1}^{n} j \cdot t_j}{\sum_{j=1}^{n} t_j}$. The denominator of this expression is simply the number of all elimination trees of $G$ of depth at most $d$ that are sensible with respect to $T$, while the numerator is the result of the modified version of $\LICZ$ where we set $\mu_j=j$ for all $j\in [n]$. Hence, we can find~$i$ (that is: pinpoint the unique root) by dividing the outcomes of two calls to weighted $\LICZ$, instead of calling $\LICZ$ $n$ times, as we did previously. Note that such division can be performed both in $\Z$ and in $\Z_p$ for any prime $p$, unless the denominator is zero. In case of $\Z_p$, it takes $\Oh(\log p)$ arithmetic operations to compute modular inverse, which unfortunately poses a technical challenge in the time complexity analysis: if applied without care, it would lead to the increase of time complexity to $\Oh(n \log n)$ time, because we would perform a linear number of divisions in $\Z_p$. This issue will be resolved in the final time complexity analysis, so let us ignore it for now.

Next, we lift the assumption about the uniqueness of the candidate for the root of an elimination tree. There are two key ingredients here. The first one is Lemma \ref{roots-num}, which bounds the number of possible candidate roots for elimination trees of optimum depth. The second one is the color coding technique. 

Suppose $\td(G)=d$. We can do so, as we enter that part of the algorithm only if $\td(G) \le d$ and we can determine $\td(G)$ by calling $\LICZ(G, T, d')$ for $d'=1, 2, \ldots, d$ and set $d$ as the smallest value of $d'$ where it returns a nonzero value, which will be equal to $\td(G)$ (assuming we did not encounter a false negative). We can also assume that $G$ is connected as otherwise we can make a separate call on each connected component. Let $R$ be the set of vertices that are potential roots of optimum-depth elimination trees of $G$; that is, $v\in R$ if and only if $\td(G-v)<d$. Then, Lemma~\ref{roots-num} implies that $|R| \in d^{\Oh(d)}$, and obviously, we have $|R| \ge 1$. Let $B\in d^{\Oh(d)}$ be the specific bound stemming from Lemma~\ref{roots-num}. Consider a random coloring of $V(G)$ with $B$ colors, that is, a~function $C \colon V(G) \to [B]$ where each vertex is independently and uniformly mapped to a~random number from $[B]$.
We note the following: (here, $e$ is the Euler's number) \begin{lemma} \label{color-coding}
	With probability at least $\frac{1}{e}$ there is a color $c \in [B]$ such that $|R \cap C^{-1}(c)| = 1$.
\end{lemma}
\begin{proof}
	Let $v$ be any vertex from $R$ (recall that $R$ is nonempty). If all other vertices from $R$ have colors different from that of $v$, then $C(v)$ is a color fulfilling the desired property. This happens with probability $$\left(1-\frac{1}{B}\right)^{|R|-1} \ge \left(1-\frac{1}{B}\right)^{B-1} = \frac{1}{\left(1 + \frac{1}{B-1}\right)^{B-1}} \ge \frac{1}{e}.\qedhere$$
\end{proof}

For each $c \in [B]$ we do the following. Create a sequence $X = (x_1, \ldots, x_n)$, where $x_i = 1$ if $C(v_i) = c$ and $x_i = 0$ otherwise, and a sequence $Y = (y_1, \ldots, y_n)$, where $y_i = i$ if $C(v_i) = c$ and $y_i = 0$ otherwise. Then, we call the modified version of $\LICZ$, where $X$ is supplied as the sequence $\mu_1, \ldots, \mu_n$, and then call it again with $Y$ instead of $X$. Similarly as in the case of unique candidates for a root from the previous paragraph, the number $i \coloneqq \frac{\sum_{j=1}^{n}t_j \cdot y_j}{\sum_{j=1}^{n}t_j \cdot x_j}$ will be the index of a possible root, provided that there exists exactly one possible root with that color. If the denominator of that expression is nonzero, $i \in C^{-1}(c)$, and $\td(G - v_i) =d-1$, then we are sure that $v_i \in R$. If we do not succeed in finding any member of $R$ for any color $c$ in this way, we repeat the procedure with a different coloring until we find one. As we execute this part of the algorithm only if $R$ is nonempty, by Lemma~\ref{color-coding}, the expected number of colorings we need to try until we discover a member of $R$ is at most~$e$.

As checking each coloring takes at most $3B \in d^{\Oh(d)}$ executions of the modified version of $\LICZ$, identifying any possible root of an optimum-depth elimination tree takes expected $2^{\Oh(d^2)} \cdot n$ time. After identifying one, we remove it from the graph, partition the remaining part into connected components (and appropriately distribute the elimination tree $T$ into elimination trees of connected components). and recurse for each connected component. After that, we connect roots of elimination trees returned from recursive calls as children of the root found on this level, obtaining an elimination tree for the whole $G$. There will be at most $d$ recursion levels and the total size of graphs on each level is at most $n$, hence the expected total work that $\LICZ$ calls will perform will be $2^{\Oh(d^2)} \cdot n$ as well. However, as mentioned before, this does not include the time needed for divisions in $\Z_p$ and we defer this analysis to a later part.

\subsection{Replacing iterative compression}

Finally, we replace the iterative compression scheme with a technique proposed by Bodlaender in his linear-time fpt algorithm to compute the treewidth of a graph~\cite{Bodlaender}. The main part of this technique was succinctly encapsulated in \cite[Lemma~2.7]{Tw5Apx}. We need a few definitions.

\begin{definition}
	For a graph $G$ and an integer $d$, the {\em{$d$-improved graph}} of $G$, denoted
	$\imp{G}{d}$, is the graph obtained from $G$ by adding an edge between every pair of vertices that are non-adjacent, but have at least $d + 1$ common neighbours of degree at most $d$ in $G$.
\end{definition}

We note the following.

\begin{lemma} \label{improved}
For every graph $G$ and integer $d$, we have $\td(G)\leq d$ if and only if $\td(\imp{G}{d})\leq d$.
\end{lemma}
\begin{proof}
    The right-to-left implication is obvious, so we need to prove that if $\td(G)\leq d$, then $\td(\imp{G}{d})\leq d$. Let $F$ be an elimination forest of $G$ of depth at most $d$. We claim that $F$ is also an elimination forest of $\td(\imp{G}{d})$. Suppose otherwise. Then there are vertices $u,v\in V(G)$ such that $\Anc_F(u,v)$ does not hold, while $uv$ is an edge in $\imp{G}{d}$. Since $F$ is an elimination forest of $G$, $u$ and $v$ are non-adjacent in $G$ but have at least $d+1$ common neighbors. However, as $\Anc_F(u,v)$ does not hold, every common neighbor of $u$ and $v$ belongs to $\tail_F(u)\cap \tail_F(v)$, which is a set of cardinality smaller than $d$. This is a contradiction.
\end{proof}

Recall that we are given an $n$-vertex graph $G$ and we would like to construct an elimination forest of $G$ of depth at most $d$, or conclude that $\td(G)>d$. It is well-known that an $n$-vertex graph of treedepth at most $d$ has at most $dn$ edges, hence we may assume that $|E(G)|\leq dn$; otherwise we immediately provide a negative answer. In that case, as proved by Bodlaender~\cite{Bodlaender}, the $d$-improved graph $\imp{G}{d}$ can be computed in time $d^{\Oh(1)} \cdot n$ using radix sort.
We call a vertex $v$ of $G$ {\em{$d$-improved-simplicial}} if the neighbourhood $N_{\imp{G}{d}}[v]$
is a clique in $\imp{G}{d}$. Note that if in $\imp{G}{d}$ there is a clique of size at least $d+1$, then $\td(\imp{G}{d})>d$, which in turn implies that $\td(G) > d$ due to Lemma~\ref{improved}.

We now recall the aforementioned statement from~\cite{Tw5Apx}.

\begin{lemma}[Lemma 2.7 of~\cite{Tw5Apx}] \label{Bod}
There is an algorithm working in time $d^{\Oh(1)} \cdot n$ time that, given an $n$-vertex graph $G$ and an integer $d$, either
\begin{enumerate}[nosep]
\item returns a maximal matching in $G$ of cardinality at least $\frac{n}{\Oh(d^6)}$, or,
\item returns a set of at least $\frac{n}{\Oh(d^6)}$ $d$-improved-simplicial vertices, or
\item correctly concludes that the treewidth of $G$ is larger than $d$.
\end{enumerate}
\end{lemma}

With the lemma stated, we are ready to optimize the way we use $\ODZYSKAJ$ in order to construct an elimination forest of $G$. 

We define a procedure $\Solve(G, d)$ that for a graph $G$ and an integer $d$, either reports that $\td(G) > d$ or provides an elimination forest of $G$ of depth at most $d$.
If $G$ consists of a single vertex, we return it as a valid elimination forest of depth $1$, so we assume that $|V(G)|>1$ from now on. As the very first step, we check if $|E(G)| \le dn$. As argued, if this is not the case, then we report that $\td(G) > d$ and terminate. Otherwise, we apply the algorithm of Lemma~\ref{Bod} with $G$ and $d$ as an input. If it reports that $\tw(G) > d$, then this implies that also $\td(G)>d$, so this conclusion can be reported and the procedure terminated.

Next, suppose the procedure returns a matching $M$ of size at least $\frac{n}{\Oh(d^6)}$. Contract all edges of $M$, thus obtaining a new graph $G_M$ as a result. Call $\Solve(G_M, d)$. Note that if this procedure returned that $\td(G_M)>d$, then we also have $\td(G)>d$, because $G_M$ is a minor of $G$ and treedepth is monotone under taking minors. Therefore, we may assume that we have obtained an elimination forest $F'$ of $G_M$ of depth at most $d$. We can now easily transform $F'$ into an elimination forest $F''$ of $G$ of depth at most $2d$, by replacing every vertex obtained from the contraction of an edge of $M$ by the two endpoints of this edge (these two vertices are put in place of the contracted as a parent and a child). Then, we may call $\ODZYSKAJ(G, F'', d)$ to either conclude that whether $\td(G) >d$, to construct an elimination forest of $G$ of depth at most $d$.

Finally, suppose the procedure of Lemma~\ref{Bod} returned a set $A$ consisting of at least $\frac{n}{\Oh(d^6)}$ $d$-improved-simplicial vertices. We compute $\imp{G}{d}$ and call $\Solve(\imp{G}{d} \setminus A, d)$. If this call reports that $\td(\imp{G}{d} \setminus A) > d$, then by Lemma~\ref{improved} we also have $\td(G)>d$, hence we can return this conclusion and terminate the algorithm. 
Otherwise, we have an elimination forest $F'$ of $\imp{G}{d} \setminus A$ of depth at most $d$. We order $A$ arbitrarily as $v_1, \ldots, v_a$ and process these vertices one by one. We shall iteratively construct $F_0, F_1, \ldots, F_a$, where each $F_i$ is an elimination forest of $G \setminus \{v_{i+1}, \ldots, v_a\}$.
We set $F_0$ to be $F$. Now, we argue how $F_i$ can be constructed from $F_{i-1}$, for $i=1,2,\ldots,a$. Since $v_i$ is $d$-improved-simplicial in $G$, the neighbourhood in $\imp{G}{d} \setminus \{v_{i+1}, \ldots, v_a\}$ is a clique (we may assume that this clique has size smaller than $d$, for otherwise it is safe to conclude that $\td(\imp{G}{d})>d$, implying $\td(G)>d$). This implies that all the neighbors of $v_i$ in this graph lie on some root-to-leaf path in $F_{i-1}$.
We can easily see that if we take the neighbor that is the lowest in $F_{i-1}$ and attach $v_i$ as its child, what we get as a result is a valid elimination forest of $\imp{G}{d} \setminus \{v_{i+1}, \ldots, v_a\}$ and we may call it $F_{i}$.
This way, we can compute $F_a$ from $F$ in time $d^{\Oh(1)} \cdot n$, and such $F_a$ is a valid elimination forest of $\imp{G}{d}$. We claim that if the depth of $F_a$ is larger than $2d$ then $\td(G)>d$. Suppose so and take any $u$ such that $\depth_{F_a}(u) = 2d+1$. Let $u_1, u_2, \ldots, u_{2d+1}$ be the path from the root to $u$ in $F_a$, where $u_{2d+1} = u$. Note that since the depth of $F$ is at most $d$, the vertices $u_{d+1}, u_{d+2}, \ldots, u_{2d+1}$ were all added in the process of obtaining $F_a$ from $F_0$, meaning that they are all $d$-improved-simplicial in $G$ and pairwise adjacent. In particular, $u_{d+1}, u_{d+2}, \ldots, u_{2d+1}$ is a clique of size $d+1$ in $\imp{G}{d}$, implying $\td(\imp{G}{d})>d$, which in turn implies that
$\td(G) > d$; so it is safe to return this conclusion then. Otherwise, we have obtained an elimination forest $F_a$ of $G$ of depth at most $2d$. It now remains to call $\ODZYSKAJ(G, F_a, d)$ to either conclude that $\td(G) \le d$, or construct an elimination forest of $G$ of depth at most $d$.

In short, the size of our graph shrinks by a constant factor with each recursive call, hence we improve the running time by a factor of $n$. We perform more detailed analysis in the next section.

\subsection{Detailed specification and the analysis of the time and space complexity}
Throughout previous subsections we introduced a series of modifications to the deterministic algorithm from Theorem~\ref{thm:main} in order to improve the $n^{\Oh(1)}$ factor to $n$. However, as there are nontrivial dependencies between these improvements and interplays between various sources of randomness, some details were omitted. Only now that we have an overall view of modifications, we may fully specify and analyze the algorithm. In this section we assume that $n$ always denotes the number of vertices of the original input graph, while $r$ denotes the number of vertices of a graph that was passed to either $\LICZ$ or $\ODZYSKAJ$ in some recursive call.

Each call of the modified version of $\LICZ$ is computed in a ring $\Z_m$ for some number $m$. If $m$ is prime, then $\Z_m$ can be equipped with a division operation so that it becomes the field $\F_m$. We promise that it will always hold that $m \in n^{\Oh(d)}$, hence the bitsizes of all numbers present in the computation will never be larger than $\Oh(d \log n)$. Hence, additions, subtractions and multiplications on such numbers take $d^{\Oh(1)}$ time and space in the RAM model.

For the unweighted version of $\LICZ$, Lemma~\ref{bound} shows the bound of $(dk \cdot 2^d)^r$ for all numbers present in the computation when performed on a graph with $r$ vertices. However, with the introduction of weights, this bound grows into $(Wdk \cdot 2^d)^r$, where $W$ is the maximum supplied weight. After the appropriate renumeration of vertices in each recursive call, we can assume that $W \le r$, which gives a bound of $(rdk \cdot 2^d)^r$ on the numbers present in the computation. Interestingly enough, even though intermediate numbers present in the computation of weighted $\LICZ$ can be as large as $r^r$, the final result $\sum_{i=1}^{r}t_i\mu_i$ can be bounded more efficiently. Namely, we have $\sum_{i=1}^{r}t_i\mu_i\leq W \sum_{i=1}^{r} t_i$ and we already know from Lemma~\ref{bound} that $\sum_{i=1}^{r}t_i \le (dk \cdot 2^d)^r$. Hence the outcomes returned by the weighted version of $\LICZ$ are bounded by $r(dk \cdot 2^d)^r$.

We need to specify what numbers $m$ we use as moduli in $\LICZ$. On one hand, we want to use large numbers, so that probabilities of errors are small. On the other hand, we need to deal with the issue of modular division cost potentially worsening our complexity to $\Oh(n \log n)$. The idea to deal with it is to distinguish two cases based on whether $r$ is large or small. If $r$ is large, the division cost will not be larger than the cost of $\LICZ$. If $r$ is small, then the bound on the result is sufficiently small so that performing the whole computation without hashing modulo a large prime (almost) fits into the RAM model and provides a true outcome at the end.

More specifically, we distinguish two cases: 
\begin{enumerate}
	\item $r \ge \log_2{n}$
	
	In that case, we use as $m$ the random prime $p$ that we drew at the beginning from the interval $(A, 2A)$, where $A = \max(L, n^5 2^{5Cd^2})$.
	We have $\log m \in d^{\Oh(1)} + \Oh(\log n)$ and the bound we use for the running time of the call to $\LICZ$ is $2^{\Oh(d^2)} r$. As a consequence, calling a modular inverse taking $\log m \cdot d^{\Oh(1)}$ time does not worsen the time complexity, as $\log m \cdot d^{\Oh(1)} \subseteq 2^{\Oh(d^2)} r$.
	
	\item $r < \log_2{n}$
	
	In that case we use $r(dk \cdot 2^d)^r+1$ as $m$. In all our calls $k = \Oh(d)$, hence numbers of this magnitude will have bitsize $\Oh(d \log n)$, so again, arithmetic operations on them can be performed in $d^{\Oh(1)}$ time in the RAM model. As explained before, even though true numbers that would be present in the computations could hypothetically exceed the value of $m$, the final result will not, hence the result modulo $m$ is equal to the true result. In other words $(\sum_{i=1}^{r}t_i\mu_i) \mod m = \sum_{i=1}^{r}t_i\mu_i$. Because of that, the division $\frac{\sum_{i=1}^{r}t_i\cdot i}{\sum_{i=1}^{r}t_i}$ can be performed on ordinary integers instead of on their moduli, and it takes $d^{\Oh(1)}$ time instead of $\Oh(\log n)$ time. We note that if this division does not result in an integer number, we already know that the we did not succeed in finding a candidate for a root in this color and we may continue to search within other colors. We also note that there is no randomness in this case, the output of this case is always correct.
\end{enumerate}

The expected total cost of divisions in the first case is not larger than the work that $\LICZ$ performs, hence it can be bounded by $2^{\Oh(d^2)} n$. Because the expected number of $\LICZ$ calls is $d^{\Oh(d)} n$, and the expected total cost of divisions in the second case is $d^{\Oh(d)} n$ as well. As such, we conclude that the expected total cost of divisions is $2^{\Oh(d^2)} n$ too. Therefore, the expected time that one $\ODZYSKAJ$ call takes on the graph on $n$ vertices is $2^{\Oh(d^2)} n$. 

In the next step, we come back to the time and space complexity analysis of the recursive scheme that replaced iterative compression technique. As for the time complexity, in both non-trivial cases we make a single recursive call on a graph with $n(1-\frac{1}{\Oh(d^6)})$ vertices, and perform additional work taking expected $2^{\Oh(d^2)} \cdot n$ time. Hence the expected time complexity $T(n, d)$ can be bounded using recurrence $$T(n, d) \le T\left(n\left(1-\frac{1}{\Oh(d^6)}\right), d\right) + 2^{\Oh(d^2)} \cdot n.$$ As in ~\cite{Bodlaender}, this recurrence solves to $T(n,d)=2^{\Oh(d^2)}\cdot n$, because unraveling the recursion results in a geometric series. As for the space complexity, we have argued that both $\ODZYSKAJ$ and internal computation of $\Solve(G,d)$ use $d^{\Oh(1)}\cdot n$ space. Therefore, the space complexity $S(n,d)$ can be bounded using recurrence
$$S(n, d) \le S\left(n\left(1-\frac{1}{\Oh(d^6)}\right), d\right) + d^{\Oh(1)} \cdot n,$$
which again solves to $S(n,d)=d^{\Oh(1)}\cdot n$.

In order to conclude, we need to bound the error probability. We recall that the randomness stemming from color coding and drawing a random prime is of type Las Vegas, that is, there is a possibility that the algorithm runs indefinitely long, but there are no errors that this randomness introduces. By using Markov's inequality we know that there is at most $\frac{1}{n}$ chance that our algorithm takes time that is at least $n$ times longer than its expected execution time, hence with at least $\frac{n-1}{n}$ probability there will be at most $2^{\Oh(d^2)} n^2$ calls to $\LICZ$. As argued before, the errors stem only from cases where the true result of $\LICZ$ should be nonzero, but becomes zero as a result of unluckily chosen modulo $m$. The probability of that happening for a particular call is at most $\frac{1}{2^{Cd^2}n^3}$ for any constant $C$ of our choice. By using the union bound, we conclude that the probability that we never encounter any error of this type is at least $\frac{n-1}{n} - \frac{2^{\Oh(d^2)} n^2}{2^{Cd^2}n^3} \ge \frac{n-2}{n}$, for any sufficiently large $C$. We remark that the errors are of the false negative type, that is, if an elimination forest is returned, it is guaranteed to a be a valid elimination forest of depth at most $d$. This concludes the description of the procedure $\Solve(G, d)$ and the analysis of its time complexity, space complexity, and the probability of correctness.
%


\section*{Acknowledgements}
We would like to thank Marcin Mucha and Marcin Pilipczuk for discussions on the topic of this work.

\bibliographystyle{abbrv}
\bibliography{ExactTD.bib}

\begin{thebibliography}{10}

\bibitem{AKS}
M.~Agrawal, N.~Kayal, and N.~Saxena.
\newblock Primes is in {P}.
\newblock {\em Annals of Mathematics}, 160, 9 2002.

\bibitem{Bodlaender}
H.~L. Bodlaender.
\newblock A linear-time algorithm for finding tree-decompositions of small
  treewidth.
\newblock {\em SIAM Journal on Computing}, 25(6):1305--1317, 1996.

\bibitem{Tw5Apx}
H.~L. Bodlaender, P.~G. Drange, M.~S. Dregi, F.~V. Fomin, D.~Lokshtanov, and
  M.~Pilipczuk.
\newblock A $c^k n$ $5$-approximation algorithm for treewidth.
\newblock {\em {SIAM} J. Comput.}, 45(2):317--378, 2016.

\bibitem{Obstructions2}
J.~Chen, W.~Czerwi\'nski, Y.~Disser, A.~E. Feldmann, D.~Hermelin, W.~Nadara,
  M.~Pilipczuk, M.~Pilipczuk, M.~Sorge, B.~Wr{\'{o}}blewski, and
  A.~Zych{-}Pawlewicz.
\newblock Efficient fully dynamic elimination forests with applications to
  detecting long paths and cycles.
\newblock In {\em 2021 {ACM-SIAM} Symposium on Discrete Algorithms, {SODA}
  2021}, pages 796--809. {SIAM}, 2021.

\bibitem{CyganNPPRW11}
M.~Cygan, J.~Nederlof, M.~Pilipczuk, M.~Pilipczuk, J.~M.~M. {van Rooij}, and
  J.~O. Wojtaszczyk.
\newblock Solving connectivity problems parameterized by treewidth in single
  exponential time.
\newblock In {\em {IEEE} 52nd Annual Symposium on Foundations of Computer
  Science, {FOCS} 2011}, pages 150--159. {IEEE} Computer Society, 2011.

\bibitem{TDESA}
W.~Czerwi\'nski, W.~Nadara, and M.~Pilipczuk.
\newblock Improved bounds for the excluded-minor approximation of treedepth.
\newblock {\em {SIAM} J. Discret. Math.}, 35(2):934--947, 2021.

\bibitem{Obstructions1}
Z.~Dvořák, A.~C. Giannopoulou, and D.~M. Thilikos.
\newblock Forbidden graphs for tree-depth.
\newblock {\em European Journal of Combinatorics}, 33(5):969--979, 2012.
\newblock EuroComb '09.

\bibitem{FurerY14}
M.~F{\"{u}}rer and H.~Yu.
\newblock Space saving by dynamic algebraization.
\newblock In {\em 9th International Computer Science Symposium in Russia, {CSR}
  2014}, volume 8476 of {\em Lecture Notes in Computer Science}, pages
  375--388. Springer, 2014.

\bibitem{GajarskyKNMPST20}
J.~Gajarsk{\'{y}}, S.~Kreutzer, J.~Ne\v{s}et\v{r}il, P.~{Ossona de Mendez},
  M.~Pilipczuk, S.~Siebertz, and S.~Toru\'nczyk.
\newblock First-order interpretations of bounded expansion classes.
\newblock {\em {ACM} Trans. Comput. Log.}, 21(4):29:1--29:41, 2020.

\bibitem{GroheK09}
M.~Grohe and S.~Kreutzer.
\newblock Methods for algorithmic meta theorems.
\newblock In {\em {AMS-ASL} Joint Special Session on Model Theoretic Methods in
  Finite Combinatorics}, volume 558 of {\em Contemporary Mathematics}, pages
  181--206. American Mathematical Society, 2009.

\bibitem{HegerfeldK20}
F.~Hegerfeld and S.~Kratsch.
\newblock Solving connectivity problems parameterized by treedepth in
  single-exponential time and polynomial space.
\newblock In {\em 37th International Symposium on Theoretical Aspects of
  Computer Science, {STACS} 2020}, volume 154 of {\em LIPIcs}, pages
  29:1--29:16. Schloss Dagstuhl --- Leibniz-Zentrum f{\"{u}}r Informatik, 2020.

\bibitem{Korhonen21}
T.~Korhonen.
\newblock A single-exponential time 2-approximation algorithm for treewidth.
\newblock In {\em 62nd {IEEE} Annual Symposium on Foundations of Computer
  Science, {FOCS} 2021}, pages 184--192. {IEEE}, 2021.

\bibitem{KowalikMNPSW20}
L.~Kowalik, M.~Mucha, W.~Nadara, M.~Pilipczuk, M.~Sorge, and P.~Wygocki.
\newblock The {PACE} 2020 {P}arameterized {A}lgorithms and {C}omputational
  {E}xperiments challenge: {T}reedepth.
\newblock In {\em 15th International Symposium on Parameterized and Exact
  Computation, {IPEC} 2020}, volume 180 of {\em LIPIcs}, pages 37:1--37:18.
  Schloss Dagstuhl --- Leibniz-Zentrum f{\"{u}}r Informatik, 2020.

\bibitem{NederlofPSW20}
J.~Nederlof, M.~Pilipczuk, C.~M.~F. Swennenhuis, and K.~W\k{e}grzycki.
\newblock Hamiltonian {C}ycle parameterized by treedepth in single exponential
  time and polynomial space.
\newblock In {\em 46th International Workshop on Graph-Theoretic Concepts in
  Computer Science}, volume 12301 of {\em Lecture Notes in Computer Science},
  pages 27--39. Springer, 2020.

\bibitem{NesetrilM08a}
J.~Ne\v{s}et\v{r}il and P.~{Ossona de Mendez}.
\newblock Grad and classes with bounded expansion {II.} {A}lgorithmic aspects.
\newblock {\em Eur. J. Comb.}, 29(3):777--791, 2008.

\bibitem{sparsity}
J.~Ne\v{s}et\v{r}il and P.~{Ossona de Mendez}.
\newblock {\em Sparsity --- Graphs, Structures, and Algorithms}, volume~28 of
  {\em Algorithms and combinatorics}.
\newblock Springer, 2012.

\bibitem{NesetrilM15}
J.~Ne\v{s}et\v{r}il and P.~{Ossona de Mendez}.
\newblock On low tree-depth decompositions.
\newblock {\em Graphs Comb.}, 31(6):1941--1963, 2015.

\bibitem{NesetrilM16}
J.~Ne\v{s}et\v{r}il and P.~{Ossona de Mendez}.
\newblock A distributed low tree-depth decomposition algorithm for bounded
  expansion classes.
\newblock {\em Distributed Comput.}, 29(1):39--49, 2016.

\bibitem{OS17}
M.~P. O'Brien and B.~D. Sullivan.
\newblock Experimental evaluation of counting subgraph isomorphisms in classes
  of bounded expansion.
\newblock {\em CoRR}, abs/1712.06690, 2017.

\bibitem{PilipczukS21}
M.~Pilipczuk and S.~Siebertz.
\newblock Polynomial bounds for centered colorings on proper minor-closed graph
  classes.
\newblock {\em J. Comb. Theory, Ser. {B}}, 151:111--147, 2021.

\bibitem{PilipczukST18}
M.~Pilipczuk, S.~Siebertz, and S.~Toru\'nczyk.
\newblock Parameterized circuit complexity of model-checking on sparse
  structures.
\newblock In {\em Proceedings of the 33rd Annual {ACM/IEEE} Symposium on Logic
  in Computer Science, {LICS} 2018}, pages 789--798. {ACM}, 2018.

\bibitem{PiWrochna}
M.~Pilipczuk and M.~Wrochna.
\newblock On space efficiency of algorithms working on structural
  decompositions of graphs.
\newblock {\em {ACM} Trans. Comput. Theory}, 9(4):18:1--18:36, 2018.

\bibitem{Reidl}
F.~Reidl, P.~Rossmanith, F.~{S{\'{a}}nchez Villaamil}, and S.~Sikdar.
\newblock A faster parameterized algorithm for treedepth.
\newblock In {\em 41st International Colloquium on Automata, Languages, and
  Programming, {ICALP} 2014}, volume 8572 of {\em Lecture Notes in Computer
  Science}, pages 931--942. Springer, 2014.

\bibitem{Primes}
J.~B. Rosser and L.~Schoenfeld.
\newblock {Approximate formulas for some functions of prime numbers}.
\newblock {\em Illinois Journal of Mathematics}, 6(1):64 -- 94, 1962.

\end{thebibliography}

\end{document}